\documentclass[aps,pra,10pt,notitlepage]{revtex4-2}
\usepackage{amsfonts,amssymb,amsmath}            % for math symbols.
\usepackage{lmodern}
\usepackage[]{graphics,graphicx}            % for graphics figures.
\usepackage{amsthm}
\usepackage{bbold,enumitem,mathtools}
\usepackage{caption}
\usepackage{tabularx}
\usepackage{subcaption}
\usepackage{tikz}
\usetikzlibrary{decorations.pathmorphing,patterns,decorations.markings,matrix,quantikz}
\usepackage{adjustbox}
\usepackage{nicematrix}
\usepackage{xstring}
\usepackage{ifthen}
\usepackage{float}
\usepackage{listings}% http://ctan.org/pkg/listings
    \usepackage{hyperref}

\usepackage[english]{babel}

\makeatletter
\def\bbl@set@language#1{%
  \edef\languagename{%
    \ifnum\escapechar=\expandafter`\string#1\@empty
    \else\string#1\@empty\fi}%
  \@ifundefined{babel@language@alias@\languagename}{}{%
    \edef\languagename{\@nameuse{babel@language@alias@\languagename}}%
  }%
  \select@language{\languagename}%
  \expandafter\ifx\csname date\languagename\endcsname\relax\else
    \if@filesw
      \protected@write\@auxout{}{\string\select@language{\languagename}}%
      \bbl@for\bbl@tempa\BabelContentsFiles{%
        \addtocontents{\bbl@tempa}{\xstring\select@language{\languagename}}}%
      \bbl@usehooks{write}{}%
    \fi
  \fi}
\newcommand{\DeclareLanguageAlias}[2]{%
  \global\@namedef{babel@language@alias@#1}{#2}%
}
\makeatother

\DeclareLanguageAlias{en}{english}

\newtheorem{theorem}{Theorem}

\newenvironment{proof*}[1][\proofname]{%
  
  \begin{proof}[#1]}{\end{proof}}

\newcommand{\identity}{\mathbb{1}}
\renewcommand{\epsilon}{\varepsilon}

\pgfset{
  foreach/parallel foreach/.style args={#1in#2via#3}{evaluate=#3 as #1 using {{#2}[#3-1]}},
}

\begin{document}

\title{Combatting the Effects of Disorder in Quantum State Transfer}
\date{\today}
\author{Catherine \surname{Keele}}
\author{Alastair \surname{Kay}}
\affiliation{Royal Holloway University of London, Egham, Surrey, TW20 0EX, UK}
\email{alastair.kay@rhul.ac.uk}
\begin{abstract}
In this paper, we examine disorder (i.e.\ static imperfections in manufacture) for the fixed-Hamiltonian evolution protocol of quantum state transfer. We improve the performance by optimising the choice of Hamiltonian, and by implementing an encoding/decoding procedure on small regions at either end of the chain. We find that encoding in only the single excitation subspace is optimal, and provides substantial enhancement to the operating regime of these systems.
\end{abstract}
\maketitle

\section{Introduction}

In the study of quantum state transfer \cite{bose2003,bose2007,kay2010a}, the aim is to specify a Hamiltonian that, if we can build it, will use its natural dynamics to transfer an unknown quantum state between two distant sites using only nearest-neighbour interactions, and benefiting from multi-particle interference to derive potential benefits such as an improvement in transfer speed over the discrete operation of quantum gates.

Inevitably, real devices will suffer from a variety of negative effects that will detract from the ideal theoretical operation. These include fabrication defects (disorder) and noise. A range of previous studies \cite{dechiara2005,kay2006b,zwick2011,ronke2016,burrell2007} have evaluated certain pre-existing solutions for the state transfer Hamiltonian, drawing conclusions about which are preferable. Aside from \cite{burgarth2005} and derivatives, which made use of multiple chains in parallel, few \cite{zwick2012} have explored the engineering problem of how best to anticipate this disorder and mitigate its effects \footnote{techniques such as error correction \cite{kay2016c,kay2017d} will work, but have not been explicitly considered in this instance.}. That is our main focus here, but we restrict to the setting where we have a single use of a single chain. We study two techniques. The first is based on the encoding/decoding technique of Haselgrove \cite{haselgrove2005} -- by encoding the state to be transferred over a small number of sites, we will see that significant gains in fidelity can be realised. Haselgrove primarily concentrated on the single excitation subspace. Here, we show that his suggested generalisation to multiple excitations is sub-optimal, and improve upon the quality of transfer that is achieved. Nevertheless, we also show that encoding in the single excitation subspace is optimal for chains with high enough transfer fidelity (coinciding with a threshold where the transfer task becomes trivial). Our second strategy is based on the idea of actively modifying the chains so that, rather than having optimal performance in the perfect, theoretical, case, their expected behaviour under a given disorder model is improved. This is primarily a numerical technique.

\subsection{State Transfer}

The model of state transfer that is typically considered makes use of the geometry of a chain to couple qubits,
\begin{equation}\label{eqn:ham}
H_0=\frac12\sum_{n=1}^NB_nZ_n+\frac12\sum_{n=1}^{N-1}J_n(X_nX_{n+1}+Y_nY_{n+1}),
\end{equation}
where $X_n$ is the standard $\sigma_x$ Pauli operator applied to qubit $n$, and $\identity$ on all other sites. Having placed an unknown quantum state $\ket{\psi}$ at one end of the chain, and initialising all other qubits in the $\ket{0}$ state, the aim is to transfer the state to the opposite end of the chain:
$$
\ket{\psi}\ket{0}^{\otimes(N-1)}\xrightarrow{e^{-iH_0t_0}}\ket{0}^{\otimes(N-1)}\ket{\psi}.
$$ 
In practice, this will never happen perfectly, and we evaluate the efficacy of the transfer using the fidelity
$$
F=\bra{\psi}\left(\text{Tr}_{1,2,\ldots,N-1}\rho\right)\ket{\psi},
$$
where $\rho$ is the evolved state of the system at the time $t_0$ that we choose to extract the state from the final qubit. As we will not be considering noise in this paper, $\rho$ will be a pure.

The Hamiltonian that we have chosen decomposes into a series of subspaces characterised by the excitation number, i.e.\ the number of 1s in the basis state. We will primarily be concentrating on the single excitation subspace comprising the basis states $\ket{n}:=\ket{0}^{\otimes (n-1)}\ket{1}\ket{0}^{\otimes (N-n)}$. Within this subspace, the Hamiltonian $H_0$ is described by an $N\times N$ matrix $H_1$ with $\{B_n\}$ on the diagonal, and $\{J_n\}$ on the sub- and super-diagonals. However, we note that the performance of the chain in higher excitation subspaces is directly related to that of the single excitation subspace thanks to the Jordan-Wigner transformation mapping to a free-fermion model \cite{jordan1928,nielsen2005}.

There are many solutions for perfect state transfer \cite{christandl2004,albanese2004,zwick2015,difranco2008}, pretty good transfer \cite{vanbommel2016,godsil2012,coutinho2016} or high fidelity transfer \cite{apollaro2012,banchi2010,bose2003}, while there are infinitely many solutions based on inverse eigenvalue problems \cite{karbach2005}. In order to facilitate fair comparison between these models, we rescale them all so that they have the same maximum coupling strength (all strengths can be rescaled $J_i\rightarrow \alpha J_i, B_i\rightarrow\alpha B_i$ provided $t_0\rightarrow t_0/\alpha$). This is a natural assumption because physical implementations will want to achieve transfer as quickly as possible (to minimise the effects of noise, whose dominant term will appear in the form $e^{-\gamma t_0}$), and hence will operate at the threshold of the largest coupling strength that can reasonably be made. The models that we consider for numerical testing are:
\begin{description}
\item [Uniform Coupling] \hfill\\
Bose's original model for state transfer \cite{bose2003} sets $J_n=1$ and $B_n=0$ for all couplings and fields. This is the simplest model, and has a fast initial transfer at a time $\sim (N+0.8 N^{1/3})/2$. However, this first peak can be weak, yielding a fidelity $F\sim\frac13+\frac16(1+1.35N^{-1/3})^2$.
\item [Apollaro Chain] \hfill\\
A simple modification of the uniform chain, setting $J_1=J_{N-1}=x$ and $J_2=J_{N-2}=y$, where $x$ and $y$ are numerically chosen to optimise the end-to-end transfer fidelity \cite{apollaro2012,banchi2010}. This method is (nearly) as fast as the uniform model, and achieves a finite fidelity of at least $F=0.99$ in the limit of long chain lengths $N$. With $y=1$ these are just the ``optimisable state transfer'' systems described in \cite{zwick2015}.
\item [PST Chain] \hfill\\
The first and most popular solution to perfect state transfer $F=1$ \cite{christandl2004,christandl2005}, with $J_n=2\sqrt{n(N-n)}/N$ (even $N$) or $J_n=2\sqrt{n(N-n)}/\sqrt{N^2-1}$ (odd $N$). This is the optimal perfect solution according to many parameters, such as state transfer time $t_0=\pi N/2$ (even) or $t_0=\pi\sqrt{N^2-1}/2$ (odd) \cite{yung2006,kay2016b}.
\item [Quadratic Chain] \hfill \\
This chain \cite{zwick2015} achieves perfect transfer, $F=1$, using a quadratic spectrum. In \cite{zwick2015}, this was shown to be particularly robust against certain types of disorder.
\end{description}
In fact, we will not consider the Quadratic Chain any further -- as we shall prove in Sec.\ \ref{sec:quadratic}, its transfer time scales as $\sim N^2$. As such, it is far more susceptible to that $e^{-\gamma t_0}$ noise term than our other candidates. Thus, even though we will not be considering noise directly in this paper, as part of the wider ecosystem, we consider this model to have been effectively eliminated already, along with the pretty good transfer variants of the uniform chain \cite{vanbommel2016,godsil2012,coutinho2016} (these studies have shown that by waiting sufficiently long, there are certain chain lengths for which arbitrarily high transfer fidelity can be achieved, but it is a long wait).

This is also the reason for eliminating what is otherwise an extremely versatile and successful model for tolerating disorder -- the dimer model \cite{wojcik2007}. In this, one takes a system (could be a chain, but need not be restricted), and adds two pendant vertices to be the qubits that are transferred between (label them 1 and $N$). These are weakly coupled to the rest of the system, and have the same magnetic field applied at a strength that is distinct from all the eigenvalues of the main system. When one analyses this from the perspective of degenerate perturbation theory, there are two eigenvectors that have support on those two vertices, which are approximately
$$
\frac{1}{\sqrt{2}}(\ket1\pm\ket{N}),
$$
with energy gap $\delta$. Hence a state starting in $\ket{1}$ arrives at site $\ket{N}$ in time $\pi/\delta$. The details of the intermediate couplings are irrelevant; if they are slightly faulty, then provided the perturbation is sufficiently small that the energies of the two pendant vertices remain distinct, the transfer still functions, it's just that the value $\delta$ might change. Since $\delta$ can be determined by measuring the actual system we have after manufacture, it can be adapted for, and a broad spectrum of disorder can be tolerated. However, the value $\delta$ is exponentially small in the minimal order of perturbation theory required, which is the distance between the two vertices (as measured by the distance on the underlying coupling graph), and will thus be incredibly small at even modest transfer distances.

\subsubsection{Transfer Time of Quadratic Model}\label{sec:quadratic}

In \cite{zwick2015}, a Hamiltonian was proposed where the spectrum is quadratic. The authors report numerical results suggesting that the transfer time scales as $\pi N^2/16$. We now give an analytic bound on this, using the techniques of \cite{yung2006,kay2016b}. For even chain lengths, we follow \cite{yung2006} directly:
$$
2J_{\max}\geq 2J_{N/2}=\text{Tr}(SH_1)=\sum\lambda_n(-1)^{n+1},
$$
where the $\lambda_n$, the eigenvalues of $H_0$ in the single excitation subspace, are ordered such that $\lambda_n<\lambda_{n+1}$ and $S$ is the swap operator
$$
S=\sum_{n=1}^N\ket{n}\bra{N+1-n}.
$$
This specific Hamiltonian is defined by the choice $\lambda_n=\pm 1,\pm 2^2,\pm 3^2,\ldots,\pm(N/2)^2$ by imposing mirror symmetry and solving an inverse eigenvalue problem. Note that the smallest gap is of size 2, so the state transfer time is $t_0=\frac{\pi}{2}$. This sum simplifies to
$$
2J_{N/2}=2(-1)^{N/2}\sum_qq^2(-1)^q=\frac{N}{4}(N+2).
$$
Thus, when we rescale such that the maximum coupling strength is 1, we have $t_0\geq\frac{\pi}{16}N(N+2)$.

For odd chain lengths, we follow \cite{kay2016b}:
$$
4J_{\max}^2\geq 4J_{(N-1)/2}^2=\text{Tr}(SH_1^2)=\sum_{n=1}^N\lambda_n^2(-1)^{n+1}.
$$
This time, the spectrum is $0,\pm1,\pm 2^2,\pm 3^2,\ldots,\pm \left(\frac{N-1}{2}\right)^2$ and the perfect transfer time will be $t_0=\pi$. We can evaluate this to find that
$$
J_{\max}^2\geq J_{(N-1)/2}^2=\frac{N+1}{64}(N^3-N^2-5N+5).
$$
As such, the perfect transfer time is at least
$$
t_0\geq \frac{\pi}{8}\sqrt{(N+1)(N-1)(N^2-5)}\sim\frac{\pi N^2}{8}.
$$
The discrepancy between even and odd cases, being roughly a factor of 2, is due to the absence of the 0 eigenvalue in the even case. In other chains, such as the PST case, this discrepancy is resolved by shifting the eigenvalues in the even length case to half-integer values. This solution is comparable to the Hahn chains with quadratic spectrum given analytically in \cite{albanese2004}, the difference being that the spectrum of \cite{albanese2004} is not symmetric about 0 and therefore contains diagonal elements in the single excitation subspace. Since state transfer for this model requires a time $O(N)$ longer than other models, we do not consider it any further.

\subsection{Disorder}

Disorder, introduced by, for example, manufacturing defects, is inevitable. Coupling strengths $J_n$ and magnetic fields $B_n$ will not be their intended strengths. There are several ways that we might parametrise these effects, with the choice being determined largely by the physical system in which we realise the state transfer.
\begin{description}
\item [Multiplicative Error] \hfill\\
In this case, $J_n\rightarrow J_n(1+\delta_n)$ where $\delta_n$ is randomly chosen from some distribution. This is highly relevant to some of the existing experiments \cite{perez-leija2013} where waveguides were positioned at a distance $r$ up to an error, and couplings were of an evanescent type, $J_n\propto e^{-\alpha r}$. Clearly this is not appropriate for magnetic fields that are all 0.
\item [Additive Error] \hfill\\
In this case, $J_n\rightarrow J_n+\delta_n$ where $\delta_n$ is randomly chosen from some distribution.
  \end{description}
One could choose any distribution for the coupling strengths. Two simple ones are the uniform distribution, where any value between $\pm\delta$ is equally likely, and the normal distribution with mean 0 and standard deviation $\sigma$. It is not expected to make a significant difference which distribution is selected. The choice between additive or multiplicative error is entirely irrelevant for the uniform chain, and is barely relevant to the Apollaro chain. On the other hand, with large variations in coupling strength in the PST chain, there may be a significant difference, and one would predict that the additive error would be far more destructive as it could easily obliterate the finely conceived coupling pattern at either end of the chain.

How should we modify our concept of fidelity in the presence of disorder? For a single instance, our previous definition of $F$ is perfectly valid. If we index different instances of disorder by $i$, then $F_i$ will describe the fidelity of an individual instance.
\begin{description}
\item [Average Fidelity] \hfill\\
If we evaluate $\bar F=\sum_{i=1}^MF_i/M$, this will evaluate what we might typically expect to achieve in a given experiment. However, it may be misleading in a regime where most values are close to 1 (as they cannot go above 1), but there can be the occasional substantial drop in fidelity (as the lower bound is 0).
\item [Minimum Fidelity] \hfill\\
The average fidelity is not much use in any given scenario. If you're provided with a chain to use, you need some guarantee on its performance. Hence a more reliable metric might be $F_{\min}=\min_iF_i$. However, we find this to be too pessimistic.
\item [Quantile Fidelity] \hfill\\
In practice, the way that one might anticipate state transfer being implemented is to manufacture several chains in advance of their being needed. We can test them all in advance and find the one with highest transfer fidelity, and use that one.  If we made $k$ samples and tested them, the probability that at least one of these is in the upper quartile (say) is $1-(\frac34)^k$. Thus, selecting a quantile (we will select the upper quartile) gives us a quantitative expectation in this multiple sample situation, and is our preferred metric.
\end{description}

\section{Encoding Strategy}

\subsection{Single Excitation Subspace}

In \cite{haselgrove2005}, Haselgrove introduced a technique for making substantial gains in state transfer fidelity. Instead of placing an unknown state $\alpha\ket{0}+\beta\ket{1}$ on the first site, and receiving it on the last site, he proposed encoding the initial state over a small set of sites $\Lambda_{\text{in}}$: $\alpha\ket{0}+\beta\ket{\Psi_{\text{in}}}$ and receiving that state on a small set of sites $\Lambda_{\text{out}}$. He concentrated primarily on encoding in the single excitation subspace such that
$$
\ket{\Psi_{\text{in}}}=\sum_{i\in\Lambda_{\text{in}}}\gamma_i\ket{i}.
$$

The method is remarkably simple -- construct the matrix
$$
M_1=\sum_{i\in\Lambda_{\text{in}}}\sum_{j\in\Lambda_{\text{out}}}\bra{j}e^{-iH_1t}\ket{i}\ket{j}\bra{i},
$$
and simply evaluate the singular value decomposition. The maximum singular value $\lambda$ is related to the transfer fidelity by
$$
F=\frac{1}{3}+\frac{(1+\lambda)^2}{6},
$$
and the optimal choice of $\ket{\Psi_{\text{in}}}$ is determined by the corresponding right singular vector, while the arriving state is given by the left singular vector. The method is also particularly relevant to the disorder scenario -- once we have manufactured a chain, we cannot control its (imperfect) couplings but we can nevertheless identify them and modify our encoding strategy based on that knowledge.

\begin{figure}
\begin{center}
\includegraphics[width=0.45\textwidth]{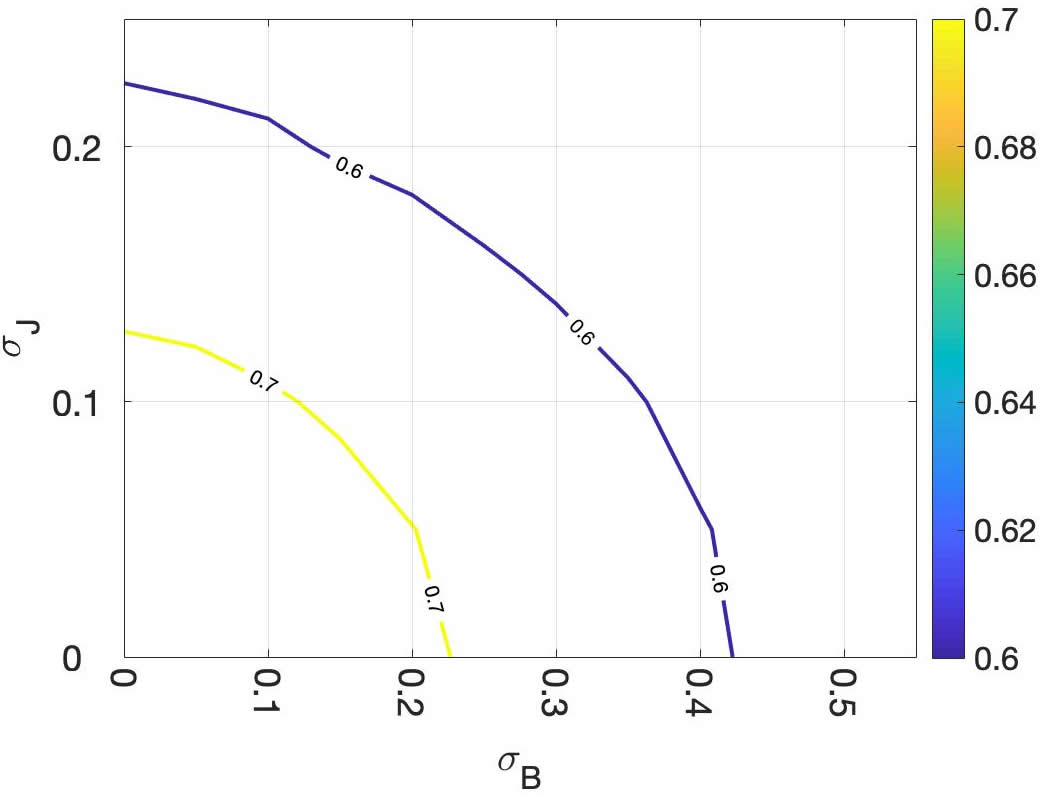}
\includegraphics[width=0.45\textwidth]{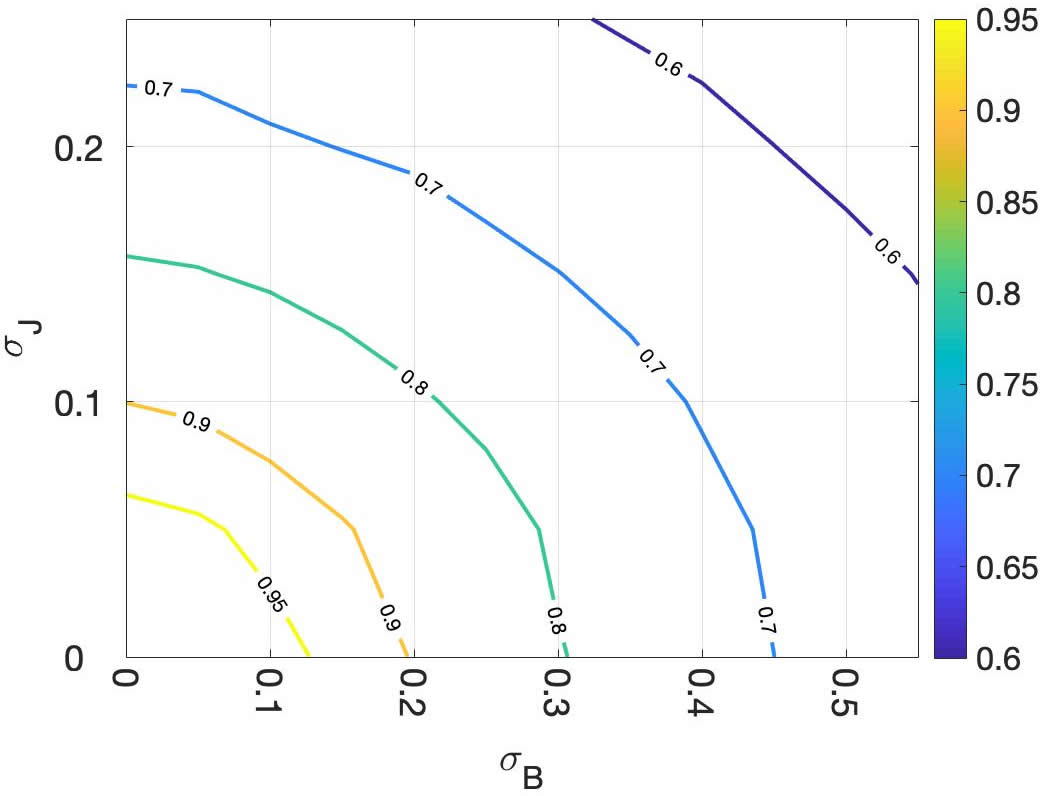}
\end{center}
\caption{Comparison of the uniformly coupled chain using no encoding (left) and encoding/decoding over 5 sites (right) in the presence of disorder. Chain length 51. Coupling (field) strength errors are selected according to a normal distribution with 0 mean and $\sigma_J$ ($\sigma_B$) standard deviation. Fidelity is the upper quartile value chosen from 1000 samples.}\label{fig:uniform}
\end{figure}

In Fig.\ \ref{fig:uniform}, we demonstrate the effect of encoding in the single excitation subspace. Even a modestly sized encoding region of 5 qubits shows a significant enhancement in performance, particularly in terms of achieving high fidelities in the weak disorder regime.

\begin{figure}
\begin{center}
\includegraphics[width=0.32\textwidth]{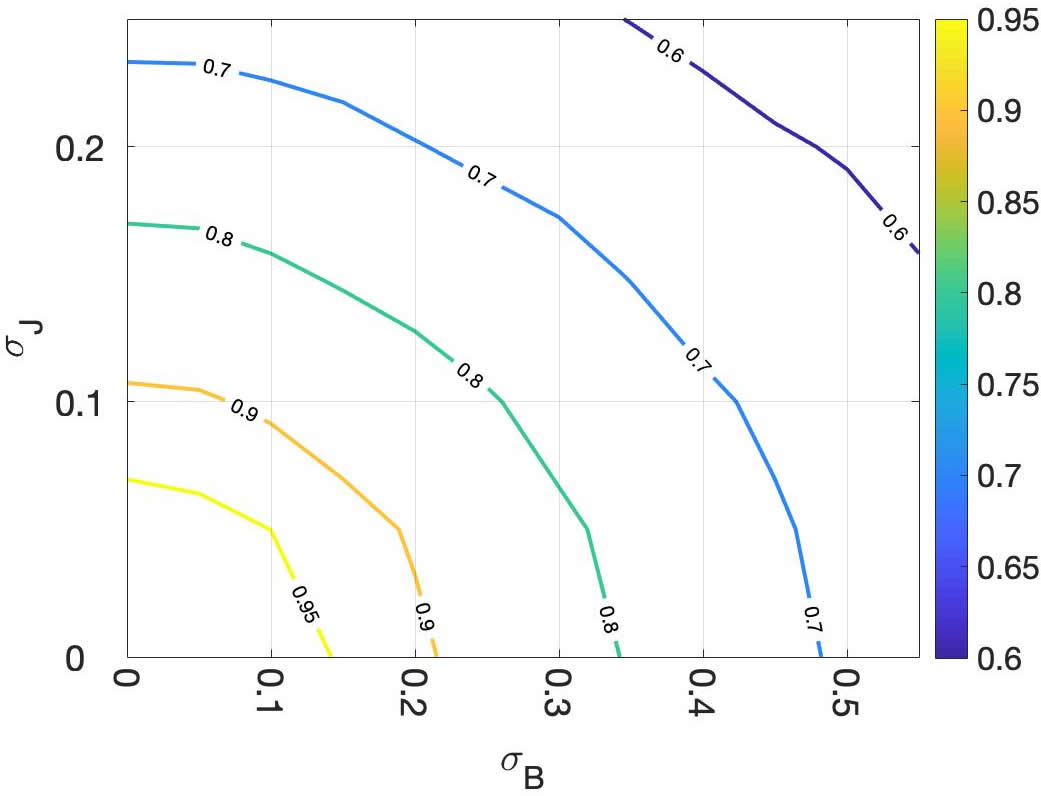}
\includegraphics[width=0.32\textwidth]{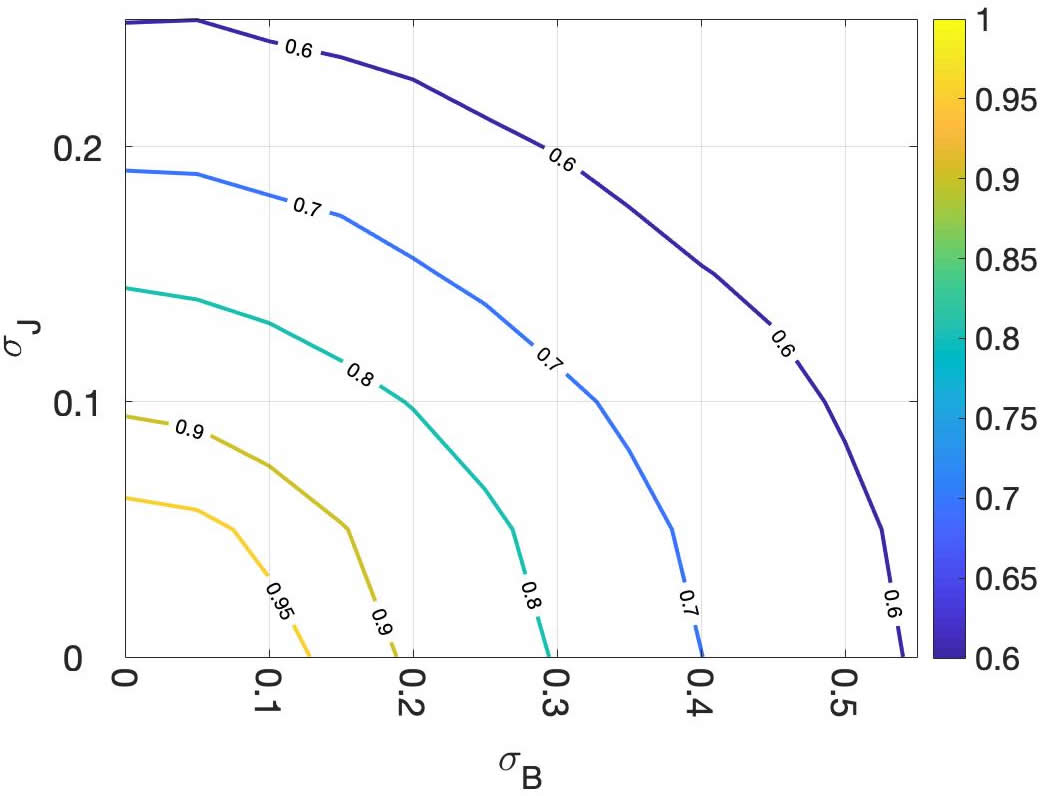}
\includegraphics[width=0.32\textwidth]{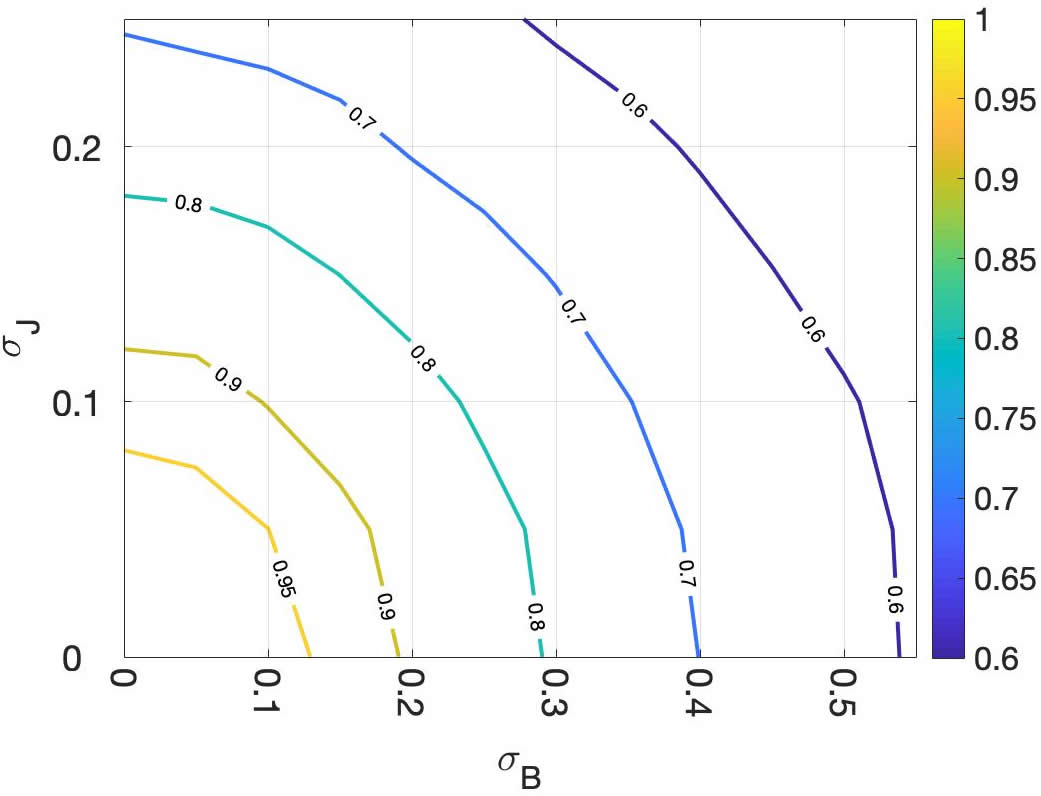}
\end{center}
\caption{Comparison of the Apollaro (left) and PST (middle, right) chains using encoding/decoding over 5 sites in the presence of disorder. Chain length 51. Coupling (field) strength errors are selected according to a normal distribution with 0 mean and $\sigma_J$ ($\sigma_B$) standard deviation. Fidelity is the upper quartile value chosen from 1000 samples, and is additive (left, middle) or multiplicative on the coupling strengths (right).}\label{fig:nonuniform}
\end{figure}

In Fig.\ \ref{fig:nonuniform}, we plot the same graph as in the right of Fig.\ \ref{fig:uniform}, with the same parameters, but for the two different chains of Apollaro and PST. We see near-identical performance to the uniform case. The similarity in performance is to be expected in the limit of large encoding sizes -- for encoding/decoding regions of size $\lceil\frac{N+1}{2}\rceil$, there is an optimal encoding of just placing the state to transfer on the central spin and using a state transfer time of 0 -- the central spin is common to both the encoding and decoding regions, so you get perfect transfer no matter what the underlying chain is. However, it is surprising to see such homogenisation of results for such a modest encoding/decoding region.

As predicted, the performance of the PST chain is worse for additive errors than it is for multiplicative errors.

\subsection{Higher Excitation Encodings}\label{sec:multiex}

Does encoding into a higher excitation subspace offer any benefit? We could directly follow Haselgrove's original paper \cite{haselgrove2005}. However, the calculation that is suggested is misleading and undervalues the transfer fidelity that is possible. Moreover, an important feature of our chosen Hamiltonian $H_0$ is that it is a free-fermion model via the Jordan-Wigner transformation \cite{jordan1928,nielsen2005}. In essence, this means that the behaviour in higher excitation subspaces is entirely determined by the behaviour in the first excitation subspace. Let us use a notation of
$$
a_n^\dagger=\frac12Z_1Z_2\ldots Z_{n-1}(X_n-iY_n),
$$
such that $a_n^\dagger\ket{0}^{\otimes N}=\ket{n}$.

Consider the right-singular vectors $\underline{u}^n$ of $M_1$. We can define
$$
b_n^\dagger=\sum_{m\in\Lambda_{\text{in}}}u^n_ma_m^\dagger.
$$
In time $t$, these evolve to
$$
c_n^\dagger=e^{-iH_0t}b_n^\dagger e^{iH_0t}=\sum_{m=1}^Nv^n_ma_m^\dagger.
$$
The vectors $\underline{v}^n=(v^n_m)_{m\in\Lambda_{\text{out}}}$ are the right singular vectors of $M_1$ up to normalisation (which is the singular value of $M_1$). This will be useful as the vectors $\underline{u}^n$ and $\underline{v}^n$ form orthonormal bases. We shall divide the output into two components, those creating excitations on the decoding region, or not:
$$
c_n^\dagger=\lambda_nc_{n,\text{out}}^\dagger+\sqrt{1-\lambda_n^2}c_{n,\overline{\text{out}}}^\dagger.
$$

Next, we define two projectors on the output region,
$$
P_0=\proj{0}^{\otimes|\Lambda_{\text{out}}|},\qquad P_1=\identity-P_0.
$$
Haselgrove suggests that the calculation we should perform is that
$$
C=\|P_1\otimes\proj{0}^{\otimes(N-|\Lambda_{\text{out}}|)}_{\overline{\text{out}}}e^{-iH_0t}\ket{\Psi_{\text{in}}}\|,
$$
and the transfer fidelity is then
$$
F=\frac13+\frac{(1+C)^2}{6}.
$$
From the Jordan-Wigner perspective, if we start with an encoding $\ket{\Psi_\text{in}}=b_1^\dagger b_2^\dagger\ket{0}^{\otimes N}$, Haselgrove is evaluating the amplitude with which \emph{both} quasi-particles arrive in the decoding region, $C=\lambda_1\lambda_2$ (as imposed by the fact that the rest of the system must be in the $\ket{0}$ state). We will now see that one can perform significantly better.

The arriving state of the chain is
$$
\alpha\ket{0}^{\otimes N}+\beta c_1^\dagger c_2^\dagger\ket{0}^{\otimes N}=\alpha\ket{0}^{\otimes N}+\beta (c_1^\dagger c_2^\dagger-\sqrt{(1-\lambda_1^2)(1-\lambda_2^2)}c_{1,\overline{\text{out}}}^\dagger c_{2,\overline{\text{out}}}^\dagger)\ket{0}^{\otimes N}+\beta\sqrt{(1-\lambda_1^2)(1-\lambda_2^2)}c_{1,\overline{\text{out}}}^\dagger c_{2,\overline{\text{out}}}^\dagger\ket{0}^{\otimes N}.
$$
This has three terms. The first is simply the initial $\ket{0}^{\otimes N}$, which, as ever, remains unchanged because it is an eigenstate of $H_0$. The second term is all the components of the $\ket{\Psi_\text{in}}$ state for which at least one excitation has arrived on the decoding region, while the third term is the component that has failed to arrive.

We now introduce a single ancilla in state $\ket{0}_A$ which we will use to receive the arriving state. To do this, we will apply some decoding unitaries. First, apply a unitary $U=P_0\otimes\identity_A+P_1\otimes X_A$. At this point, however, the chain and the ancilla are highly entangled. If we then apply a controlled-unitary $c-V$, controlled off the ancilla and targeting the decoding region of the chain, we can partially disentangle the two systems. We retain significant freedom to choose $V$ to maximise the fidelity. After the controlled-unitaries, the state, now including ancilla, is
$$
\alpha\ket{0}^{\otimes N}\ket{0}_A+\beta V(c_1^\dagger c_2^\dagger-\sqrt{(1-\lambda_1^2)(1-\lambda_2^2)}c_{1,\overline{\text{out}}}^\dagger c_{2,\overline{\text{out}}}^\dagger)\ket{0}^{\otimes N}\ket{1}_A+\beta\sqrt{(1-\lambda_1^2)(1-\lambda_2^2)}c_{1,\overline{\text{out}}}^\dagger c_{2,\overline{\text{out}}}^\dagger\ket{0}^{\otimes N}\ket{0}_A.
$$
Tracing out the chain leaves a mixed state
\begin{multline*}
\proj{0}(|\alpha|^2+|\beta|^2(1-\lambda_1^2)(1-\lambda_2^2))+\proj{1}|\beta|^2(1-(1-\lambda_1^2)(1-\lambda_2^2))\\+\lambda_1\lambda_2(\alpha\beta^*\ket{0}\bra{1}\bra{0}c_{2,\text{out}}c_{1,\text{out}}V^\dagger \ket{0}+\alpha^*\beta\ket{1}\bra{0}\bra{0}Vc_{1,\text{out}}^\dagger c_{2,\text{out}}^\dagger\ket{0})
\end{multline*}
If we select $V$ such that $Vc_{1,\text{out}}^\dagger c_{2,\text{out}}^\dagger\ket{0}=\ket{0}$, this clearly serves to maximise the transfer fidelity and we get a total fidelity (averaged over all possible input states) of
$$
\frac13+\frac{(1+\lambda_1\lambda_2)^2}{6}+\frac{1}{6}\left(1-\lambda_1^2\lambda_2^2-(1-\lambda_1^2)(1-\lambda_2^2)\right).
$$
For multiple excitations, this generalises to
$$
\frac13+\frac{(1+\prod_i\lambda_i)^2}{6}+\frac{1}{6}\left(1-\prod_i\lambda_i^2-\prod_i(1-\lambda_i^2)\right),
$$
where the third term, which is non-negative, is the enhancement over the result of \cite{haselgrove2005}. We interpret this function as the success probability of the entire state transferring as we would wish (the first two terms, as predicted by \cite{haselgrove2005}) plus some additional terms adding to the weight of the arrival of $\ket{1}$, but not contributing to the coherence with the $\ket{0}$ term. These are all the terms except for the perfectly arriving state (which we have already counted) and the term for which none of the excitations arrive on the output region.

It is worth nothing that in the case of end-to-end transfer, this decoding process automatically incorporates that `transfer phase' which is often removed in a more \emph{ad hoc} manner.

\subsection{Choosing the Best Subspace for Encoding}

If it were the case that adding an excitation always increased the fidelity, then the optimal encoding would always be the all-ones state, independent of chain, and the encoding method would be relatively simple. This is not generally the case.

\begin{theorem}
For any state transfer protocol utilising a Hamiltonian of the form $H_0$ given in Eq.\ (\ref{eqn:ham}), if the largest singular value of $M_1$ is at least $\sqrt{2}-1$, then optimal performance is achieved by encoding in the single excitation subspace.
\end{theorem}
\begin{proof}
Let us denote by $F_n$ the fidelity achieved by encoding in $n$ excitations (using the $n$ largest singular values of $M_1$). If the values $\lambda_1$ to $\lambda_n$ are fixed, how are we best to select $\lambda_{n+1}$ under the constraint $0\leq\lambda_{n+1}\leq\lambda_{n}$? Consider the enhancement in fidelity by including this excitation:
$$
\Delta=F_{n+1}-F_n=\frac16\left(2(\lambda_{n+1}-1)\prod_{i=1}^n\lambda_i+\lambda_{n+1}^2\prod_{i=1}^n(1-\lambda_i^2)\right).
$$
The derivative is
$$
\frac{\partial \Delta}{\partial\lambda_{n+1}}=\frac16\left(\prod_{i=1}^n\lambda_i+2\lambda_{n+1}\prod_{i=1}^n(1-\lambda_i^2)\right),
$$
which is clearly positive. In other words, the fidelity is greatest by setting $\lambda_{n+1}=\lambda_n$.

Let us then proceed by setting all $\lambda_{n}=\lambda_1=\lambda$ such that
\begin{align*}
F_n&=\frac16\left(4+2\lambda^n-(1-\lambda^2)^n\right)\\
\Delta&=\frac16\lambda^2(1-\lambda)\left(-2\lambda^{n-2}+(1+\lambda)(1-\lambda^2)^{n-1}\right).
\end{align*}
We shall break our proof into a series of ranges.

Consider the range $\lambda\leq 1-\lambda^2\leq 2\lambda$. We directly evaluate
$$
6(F_n-F_1)=2(\lambda^{n-1}-1)\lambda-((1-\lambda^2)^{n-1}-1)(1-\lambda^2).
$$
Using the upper range, we have
$$
6(F_n-F_1)\leq 2\lambda(\lambda^{n-1}-(1-\lambda^2)^{n-1}).
$$
This is non-positive given the lower limit of the range. Hence, in this range, there is no benefit in using multiple excitations.

Next, consider the range $\lambda\geq 1-\lambda^2$. Since $(1+\lambda)(1-\lambda^2)<2$,
\begin{align*}
\Delta_n&=\frac16\lambda^2(1-\lambda)((1+\lambda)(1-\lambda^2)^{n-1}-2\lambda^{n-2}) \\
&\leq \frac13\lambda^2(1-\lambda)((1-\lambda^2)^{n-2}-\lambda^{n-2}).
\end{align*}
Since $(1-\lambda^2)^{n-2}<\lambda^{n-2}$, this is, again, negative. Thus, $F_n-F_1=\sum_k\Delta_k$ is negative, and $F_1$ is the largest fidelity.

\end{proof}
In other words, for any system that has sufficiently high state transfer fidelity, the single excitation subspace is optimal for encoding.

If we have $\lambda\leq \sqrt{2}-1$, this means that $F_n\leq\frac23$ for all $n$. Since $F=\frac23$ is the classical threshold for state transfer (measure the qubit, send the measurement result, and recreate the measured state), we would never be interested in operating below this threshold for state transfer (entanglement transfer, see Sec.\ \ref{sec:entdist} may still operate in this regime). In other words, we should always use the single excitation subspace for encoding. This means that, tuned to the specific instance of disorder, these single-excitation encodings must out-perform general purpose error correcting codes of the same size, such as \cite{kay2016c,kay2017d}.

\section{Hamiltonian Modification}

Rather than merely assessing how well existing solutions for state transfer perform, a stronger target would be to find a solution that has optimal performance in the presence of disorder. An attempt to find the globally optimal solution would be extremely challenging. Unconstrained, we anticipate that the dimer solution would likely be the solution but, as already specified, we need to constrain parameters such as the transfer time (particularly its scaling) and maximum coupling strength.

\subsection{Perturbations}

Instead of a globally optimal solution, perhaps we can find locally optimal solutions? Consider an initial Hamiltonian $H_0$, which suffers disorder in the form of perturbations $\delta H$. Can we add a new perturbation $V$ that compensates for the average effect of $\delta H$? (This could be specific to the disorder model parameters.) To this end, let us consider the Dyson expansion of the time evolution operator for end-to-end transfer
\begin{multline*}
f=\bra{N}e^{-i(H_0+\delta H+V)t_0}\ket{1}=\bra{N}e^{-iH_0t_0}\ket{1}-i\int_0^{t_0}dt\bra{N}e^{-i(H_0+V)(t_0-t)}\delta He^{-i(H_0+V)t}\ket{1}\\-\frac12\int_0^{t_0}dt\int_0^tdt_2\bra{N}e^{-i(H_0+V)(t_0-t)}\delta He^{-i(H_0+V)(t-t_2)}\delta He^{-i(H_0+V)t_2}\ket{1}+O(\delta^3).
\end{multline*}
For simplicity, denote these three terms by $f_0$, $f_1$ and $f_2$ respectively. The fidelity is
$$
F=|f|^2=|f_0|^2+f_0f_1^*+f_0^*f_1+f_0f_2^*+f_0^*f_2+|f_1|^2.
$$
Note that $f_0$ and $f_1$ do not depend on $\delta H$ beyond first order. When we average over $\delta H$, it must be that the terms $f_0f_1^*+f_0^*f_1$ vanish because our chosen distribution has 0 mean. Hence, to first order, $\bar F=|f_0|^2$. It would thus appear that we are best (for sufficiently weak disorder) to work with perfect transfer chains. This is hardly surprising. Let us therefore take $H_0+V=H_P$, a perfect transfer chain with state transfer time $t_0$. We will make the further assumptions that our Hamiltonian is field-free and that $N$ is odd, both consistent with all the models we have considered so far. When we revisit the Dyson expansion, this guarantees that
$$
f_0=\bra{N}e^{-iH_Pt_0}\ket{1}=\pm 1,
$$
and that
$$
f_1=-i\int_0^{t_0}dt\bra{N}e^{-iH_P(t_0-t)}\delta He^{-iH_Pt}\ket{1}=\mp i\int_0^{t_0}dt\bra{1}e^{iH_Pt}\delta He^{-iH_Pt}\ket{1}.
$$
As $\delta H$ is Hermitian, $\bra{\psi}\delta H\ket{\psi}$ is always real. Hence, $f_0f_1^*$ is imaginary, but $f_0f_1^*+f_0^*f_1$ only selects the real component. In other words, by selecting a perfect transfer Hamiltonian, we are guaranteed that the effect of disorder is $O(\delta H^2)$, not only after averaging over all disorder, but for any individual case of disorder.

Let us briefly attempt to motivate that this is not the generic case -- a Hamiltonian without perfect transfer will have terms $O(\delta H)$. To see this, let us write $e^{-iHt_0}\ket{N}=\ket{\psi}$. We continue to assume that the (unperturbed) Hamiltonian is field-free, and that $N$ is odd. Hence $\braket{n}{\psi}$ is real for odd $n$ and imaginary for even $n$ \footnote{This follows from the fact that $DH_1D=-H_1$ where $D=\sum_n(-1)^{n+1}\proj{n}$, which allows us to relate $\braket{n}{\psi}$ and $\braket{n}{\psi}^*$.}. We will give two analytic (but imperfect) cases:
\begin{itemize}
\item If $\braket{2}{\psi}\neq 0$, then select $\delta H=\epsilon H_0$ for some small $\epsilon.$ Since $H_0$ commutes with $e^{-iH_0t}$, we get
$$
F=|\braket{1}{\psi}|^2+2\epsilon\text{Re}\left(\braket{\psi}{1}\int_0^{t_0}dt\bra{1}e^{-iH_0t}H_0e^{iH_0t}\ket{\psi}\right)=|\braket{1}{\psi}|^2+2\epsilon t_0\braket{\psi}{1}\braket{2}{\psi}\bra{1}H_0\ket{2},
$$
which clearly has an $O(\epsilon)$ term. However, this is imperfect because if $t_0$ has been selected to optimise the transfer fidelity,
$$
\frac{d}{dt}\bra{1}e^{-iH_0t}\ket{1}=0\implies \braket{2}{\psi}=0,
$$
contradicting this assumption.
\item Alternatively, let $\delta H=\epsilon H_0^3$. Then, we rely on $\braket{4}{\psi}\neq 0$ which is not constrained in the same way as $\braket{2}{\psi}$. However, this uses disorder that is not a modification of existing coupling terms.
\end{itemize}
These imperfect cases are suggestive but not absolute. Instead, we resort numerics. For example, the uniform chain of length 51 with $t_0$ coinciding with the first transfer peak, with a perturbation $\epsilon(\ket{1}\bra{2}+\ket{2}\bra{1})$ has a first order correction to the fidelity of $\sim 0.017\epsilon$. This is further supported in the left-hand panel of Figure \ref{fig:xy}. Close to the optimal choice of parameter for the Apollaro model, we see ellipses of constant fidelity, indicating a linear regime.

Perfect state transfer chains are local optima in the space of Hamiltonians in respect of resistance to Hamiltonian perturbations. As perfect transfer Hamiltonians are characterised by a discrete spectral property \cite{kay2010a}, there is not opportunity to move smoothly within this space to achieve further optimisation.

\subsection{Apollaro Revisited}

If we wish to work in a regime where time is so short that there aren't perfect transfer solutions, there are a limited number of high transfer fidelity models that we can start from. We consider the Apollaro model, with its parameters $x$ and $y$. Those values of $x$ and $y$ were initially chosen to optimise the transfer fidelity for end-to-end transfer in the absence of disorder. For a chain of length 51, for example, \cite{apollaro2012} gives the optimal values. But now we are optimising according to different criteria, whether that be maximising the transfer fidelity in the absence of disorder but making use of encoding across multiple sites, or additionally introducing disorder (much as in \cite{zwick2015} without encoding). One should expect that the optimal values would be different. We have numerically assessed these in the case of $N=51$ and present the results in Table \ref{tab:xy}. There is substantial variation in optimal choice of $x$ and $y$ depending on the various parameters of disorder strength and number of encoding/decoding sites, indicating the need for a case by case optimisation in order to achieve the peak fidelity. The variation with parameters $x,y$ in the upper-quartile transfer fidelity for a specific case of disorder is shown in the right-hand panel of Fig.\ \ref{fig:xy}. While there is benefit to optimising for the specific situation, Fig.\ \ref{fig:xy} shows that the main peak is broad, meaning that the model is quite permissive.

\begin{table}
\begin{tabularx}{5.5cm}{>{\setlength\hsize{1\hsize}\centering}Xc|cc}
size of\newline encoding region & $\delta$ & $x$ & $y$ \\\hline\hline
1 & 0 &  0.4322 & 0.7338 \\
1 & 0.1 &0.43& 0.74 \\
3 & 0 & 0.48 & 0.8 \\
3 & 0.1 & 0.50 & 0.84 \\
3 & 0.15 & 0.53 & 0.91 \\
5 & 0.1 & 0.38 & 0.76
\end{tabularx}
\caption{Optimal choice of parameters for Apollaro chain in different settings, chain length $N=51$. Disorder comprises absolute errors from a uniform distribution in range $\pm \delta$ on coupling strengths ($J_n$) only.}\label{tab:xy}
\end{table}

\begin{figure}
\includegraphics[width=0.45\textwidth]{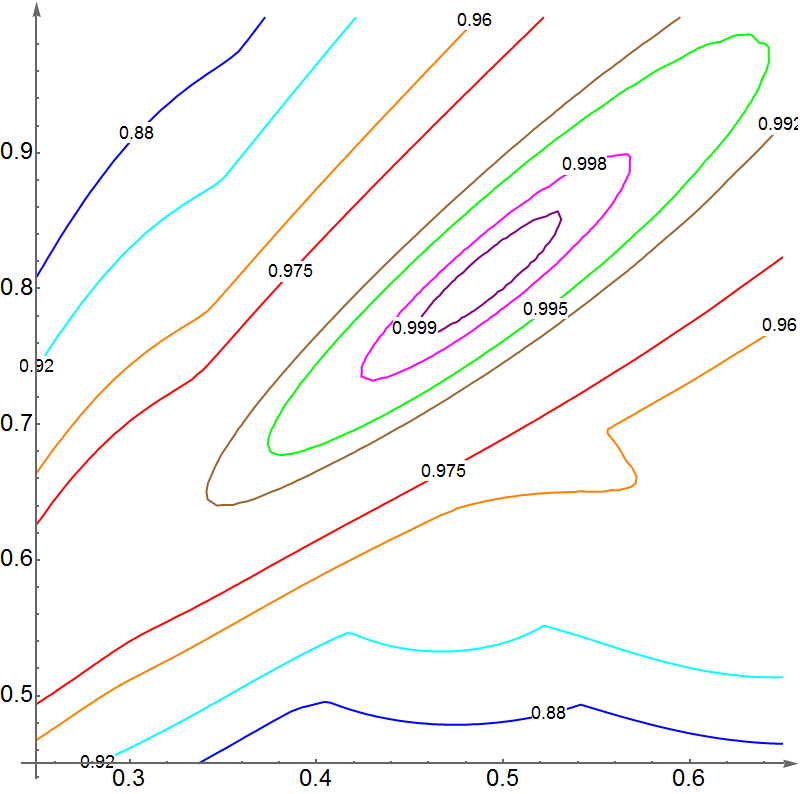}
\includegraphics[width=0.45\textwidth]{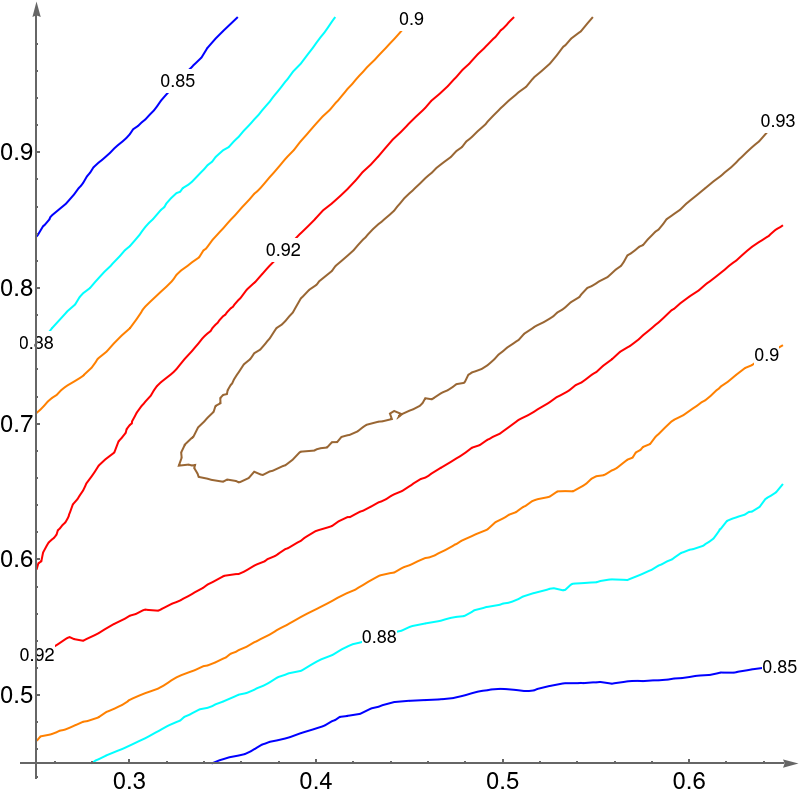}
\caption{Plots of achieved fidelity of $x$ vs.\ $y$ parameters in Apollaro model, encoding/decoding over 3 sites, chain length 51. Left: no disorder. Right: Absolute errors of up to 0.1 on coupling strengths ($J_n$) only (uniform distribution, 0 mean).}\label{fig:xy}
\end{figure}

\section{Entanglement Distribution}\label{sec:entdist}

State transfer chains are also useful for entanglement distribution. Instead of using a single unknown qubit state as input, you supply one half of a Bell pair. This half gets transferred to the opposite end of the chain, and you have an approximate Bell pair shared between two distant parties. Clearly, this protocol can be updated to incorporate encoding and multiple excitations (whether that's the creation of a single Bell pair using a multiple excitation encoding, or multiple Bell pairs). The purpose of this section is to argue that there is essentially no benefit to anything other than using a single-excitation encoding.

To see that it is sufficient to restrict to the distribution of a single Bell pair using the single-excitation subspace, consider the following protocol for transferring an unknown state:
\begin{itemize}
\item Initialise the whole chain in the $\ket{0}^{\otimes N}$ state.
\item Create $n$ Bell pairs on $2n$ qubits (not on the chain).
\item Transfer the $n$ qubit state comprising one half of each Bell pair onto the first $k\geq n$ qubits of the chain. This transfer may involve a transformational unitary $U$ that implements an encoding which we can optimise over.
\item Perform the state transfer protocol.
\item Decode the last $k$ qubits of the chain onto $n$ ancilla qubits.
\item Create, on $2n$ qubits, a single logical Bell pair where each logical qubit is encoded into $n$ physical qubits of the best possible error correcting code (whatever that might be).
\item Teleport one half of that Bell pair through the distributed Bell pairs that we created.
\item Decode the logical qubits at either end (incorporating error correction). This leaves us with a single, high fidelity Bell pair shared between the two parties at opposite ends of the chain.
\item Teleport an unknown quantum state from one party to another. This arrives with fidelity $F_\text{teleport}$.
\end{itemize}
Note that all the teleportation operations (i.e.\ measurements) can be applied by the party at the start of the chain. With respect to the rest of the protocol, it does not matter \emph{when} these measurements are made. Thus, instead, consider that these measurements are made \emph{before} the state transfer stage, at which point it is clear that this is just a normal state transfer protocol using an encoding unitary $U$ over $k$ qubits, and yet the fidelity of transfer is $F_\text{teleport}$, which therefore cannot exceed the optimal state transfer fidelity \footnote{We are talking specifically here about the state transfer fidelity averaged over all possible input states, not just the fidelity of transfer of a single excitation. This is important due to the influence of the corrective unitaries in the teleportation protocol, which are not applied until after the arrival of the state, and have an averaging effect because, for example, the corrective unitaries for a single-qubit teleportation are the Pauli operators, which form a 2-design \cite{dankert2009}.}, which was created by using a single excitation subspace encoding, as demonstrated in Sec.\ \ref{sec:multiex} (assuming the original transfer chain is of sufficiently high quality).

Teleportation, in this instance, provides no enhancement. Its benefits arise from repeated use of the chain to transfer many Bell pairs independently. These can then be distilled into a single high quality Bell pair, allowing chains with particularly weak fidelities (below the $F=2/3$ classical threshold) to still achieve high quality transfer. Each individual usage is still just as sensitive to disorder as any state transfer protocol, and is just as responsive to the techniques described through this paper, particularly encoding, for enhancing its performance. This may mean that we wish to use chains in the region where multiple excitation encoding could be useful. Nevertheless, explicit solution of the equations shows that there is no advantage unless the encoding and decoding regions contain at least 13 qubits. %That said, it is more likely that one might work in the regime of lower quality transfer for an individual use, which makes it more likely that encoding into higher excitation subspaces could become useful.

\section{Conclusions}

Our ultimate conclusion is much the same as previous work \cite{zwick2015} -- that once disorder is taken into account, there is little to choose between various high-fidelity state transfer models. However, the key difference is that by using even modestly sized encoding and decoding regions, we can massively enhance the transfer fidelities, to the extent that even the uniformly coupled chain becomes competitive. When considering other parameters relative to which we might like to optimise (ease of manufacture, transfer speed to minimise the effects of noise), this is incredibly useful.

We have proven that encoding into the single excitation subspace is optimal for all chains worth considering for state transfer. This makes a crucial difference to the ability to calculate optimal encodings because the computation required just involves the single excitation subspace, and is therefore a computation on an $N\times N$ matrix for a system of $N$ qubits, rather than a prohibitive calculation on the full $2^N$-dimensional space. This result applies specifically to chains of the XX type, for which the Jordan-Wigner transformation is applicable. It would be interesting to understand if it extends beyond that scope to incorporate the XXZ coupling model, including the Heisenberg model as a special case.

Additionally, we showed that the solutions for perfect state transfer are locally optimal against perturbations, i.e.\ small disorder. At shorter state transfer times, models such as that of Apollaro \cite{apollaro2012} can be specifically tuned to give improved robustness.

The relevance of our results to current experimental implementations depends on the platform. Much focus has recently been placed on superconducting systems, such as those from IBM and Google \cite{arute2019}. While these provide a fixed network of qubits, coupled by exactly the type of Hamiltonian discussed here, the publicly available devices are not built with an aim of producing a specific set of coupling strengths, or minimising the error in those parameters. Moreover, the energy scales of the magnetic fields are a different order of magnitude compared to the coupling strengths. This by itself is not a problem; any uniform magnetic field is tolerated. However, it tends to mean that the error scale of the magnetic field is the same strength as the coupling strengths themselves (i.e.\ on the plots, one should be looking in the regime $\sigma_B \gtrsim1$). This remains some distance from utility without encoding regions on the order of the device size, and is potentially susceptible to Anderson Localisation \cite{ronke2016}. Other experimental systems are far more promising. For instance, optical passage \cite{chapman2016,perez-leija2013} has errors that are on a far more reasonable scale. In \cite{perez-leija2013}, the error in coupling strength $J$ is estimated to be a multiplicative error of about 10\%, with no significant field error. This could undoubtedly be improved, but places it firmly in a regime that would demonstrate substantial benefit from encoding and decoding.

\bibliography{../../../References}
\end{document}